\documentclass[runningheads]{llncs}

\usepackage[T1]{fontenc}
\usepackage[utf8]{inputenc}
\usepackage{amsmath}
\usepackage{amssymb}
\usepackage{multicol}
\usepackage{xspace}
\usepackage{upgreek}
\usepackage[textsize=scriptsize]{todonotes}
\usepackage{algorithm}

\usepackage[noEnd=true]{algpseudocodex}
  \algnewcommand{\IIf}[1]{\State\algorithmicif\ #1\ \algorithmicthen}
  \algnewcommand{\EndIIf}{\unskip\ \algorithmicend\ \algorithmicif}

  \makeatletter
  \renewcommand{\ALG@beginalgorithmic}{\scriptsize}
  \makeatother
  \algrenewcommand\alglinenumber[1]{\tiny #1:}

\usepackage{proof}

\usepackage[draft]{fixme}

\usepackage[font=small,skip=8pt]{caption}
\usepackage{stmaryrd}
\usepackage{paralist}
\usepackage{proof}

\usepackage{tikz}
\usetikzlibrary{positioning}

\usepackage{multirow}

\usepackage{xparse}

\usepackage[all]{xy}
\newdir{|>}{%
  !/4.5pt/@{|}*:(1,-.2)@^{>}*:(1,+.2)@_{>}}

\usepackage[english]{babel}

\usepackage{hyperref}
\hypersetup{
  linkcolor  = blue,
  citecolor  = blue,
  urlcolor   = blue,
  colorlinks = true,
}

\usepackage{mathrsfs}

\usepackage{mathtools}
\usepackage{forest}
\usepackage{etoolbox}

\usepackage{thm-restate}

\usepackage{cleveref}

\Crefname{algorithm}{Alg.}{Algs.}
\Crefname{definition}{Def.}{Defs.}
\Crefname{figure}{Fig.}{Figs.}
\Crefname{proposition}{Prop.}{Props.}
\Crefname{theorem}{Thm.}{Thms.}
\Crefname{example}{Ex.}{Exs.}
\Crefname{corollary}{Cor.}{Cors.}
\Crefname{section}{Sec.}{Secs.}
\Crefname{appendix}{App.}{Apps.}

\usepackage{tikz}
\usetikzlibrary{arrows,decorations,shapes,automata,positioning,decorations.pathmorphing,patterns}

\tikzset{
  every picture/.style = {
    thick,
    ->,
    >=stealth',
    node distance = 1.5em and 3em,
  }
  ,
  cross line/.style = {
    preaction = {
      draw=white,
      -,
      line width=4pt
    }
  }
  , 
  state/.style = {
    rectangle,
    rounded corners = 5pt,
    font = \footnotesize,
    draw,
    minimum width = 1em,
    minimum height = 1em
  }
  , 
  label-state/.style = {
    sloped,
    font = \scriptsize,
    label distance = -2pt
  }
  , 
  label-edge/.style = {
    font = \scriptsize,
    label distance = -2pt
  }
}

\newcommand{\PROP}{{\rm \sf Prop}\xspace}
\newcommand{\ACT}{{\rm \sf Act}\xspace}

\newcommand{\kh}{{\mathsf{Kh}}}

\newcommand{\lra}{\leftrightarrow}

\newcommand{\liff}{\leftrightarrow}

\newcommand{\A}{{\operatorname{\sf A}}}
\newcommand{\E}{{\operatorname{\sf E}}}

\newcommand{\SAT}{\mathsf{Sat}}

\newcommand{\KHlogic}{\ensuremath{\mathsf{L}_{\kh}}\xspace}

\newcommand{\model}{\mathfrak{M}}

\newcommand{\R}{\operatorname{R}}
\renewcommand{\S}{\operatorname{S}}

\newcommand{\V}{\operatorname{V}}

\newcommand{\plan}{\pi}

\newcommand{\PLANS}{\ACT^{*}}

\newcommand{\lts}{\textup{LTS}\xspace}

\newcommand{\truthset}[2]{\llbracket #2 \rrbracket^{#1}}

\newcommand{\stexec}{\operatorname{SE}}

\newcommand{\card}[1]{{\mid}{#1}{\mid}}
\newcommand{\tup}[1]{\langle{#1}\rangle}

\newcommand{\csetsc}[2]{\{{#1}\mid {#2}\}}
\newcommand{\setof}[2]{\{{#1}\mid {#2}\}}
\newcommand{\set}[1]{\{{#1}\}}

\newcommand{\colornuevo}{teal}

\newcommand{\colornota}{Peach}

\newcommand{\colorincompleto}{red}

\newcommand{\NP}{{\rm\textsf{NP}}\xspace}
\newcommand{\CoNP}{{\rm\textsf{Co-NP}}\xspace}

\newcommand{\Poly}{{\rm\textsf{P}}\xspace}
\newcommand{\PSPACE}{{\rm\textsf{PSpace}}\xspace}

\newcommand{\PH}{{\rm\textsf{PH}}\xspace}

\DeclareMathOperator{\sat}{\SAT}
\DeclareMathOperator{\unsat}{\mathsf{Unsat}}

\newcommand{\tset}[1]{\llbracket #1 \rrbracket}

\newcommand{\zerodisplayskips}{%
  \setlength{\abovedisplayskip}{5pt}%
  \setlength{\belowdisplayskip}{5pt}%
  \setlength{\abovedisplayshortskip}{5pt}%
  \setlength{\belowdisplayshortskip}{5pt}}
\appto{\normalsize}{\zerodisplayskips}
\appto{\small}{\zerodisplayskips}
\appto{\footnotesize}{\zerodisplayskips}

\DeclareMathOperator{\depth}{\mathsf{md}}
\DeclareMathOperator{\sforms}{\mathsf{sf}}

\DeclareMathOperator{\C}{\mathrm{\Pi}}

\begin{document}
\title{How Easy it is to Know How: \\
An Upper Bound for the Satisfiability Problem}
\titlerunning{How Easy it is to Know How}
\author{
Carlos Areces \inst{1,2} \and
Valentin Cassano \inst{1,2,3} \and
Pablo F. Castro \inst{1,3} \and Raul Fervari \inst{1,2,4} \and
Andr\'es R. Saravia \inst{1,2}}
\authorrunning{C. Areces et al.}
%
\institute{Consejo Nacional de Investigaciones Cient\'ificas y T\'ecnicas (CONICET), Argentina \and Universidad Nacional de Córdoba (UNC), Argentina \and Universidad Nacional de Río Cuarto (UNRC), Argentina \and Guangdong Technion - Israel Institute of Technology (GTIIT), China}

\maketitle

\begin{abstract}
    We investigate the complexity of the satisfiability problem for a modal logic expressing `knowing how' assertions, related to an agent's abilities to achieve a certain goal. 
    We take one of the most standard semantics for this kind of logics based on linear plans.
    Our main result is a proof that checking satisfiability of a `knowing how' formula can be done in $\Sigma_2^P$.
    The algorithm we present relies on eliminating nested modalities in a formula, and then performing multiple calls to a satisfiability checking oracle for propositional logic.

    \keywords{Knowing How  \and Complexity \and Satisfiability.}

    \bigskip
    {\color{red} {\bf Errata}: This is a complete version of our paper 
published in the 18th Edition of the European Conference on Logics in 
Artificial Intelligence (JELIA~2023) hosted by the TU Dresden, Germany, in 
Sep.~20---22, 2023. This version corrects the published version since the 
latter is missing the algorithms for the satisfiability problem we studied.}
\end{abstract}

\section{Introduction}
\label{sec:intro}
The term `Epistemic Logic'~\cite{Hintikka:1962} encompasses a family of logical formalisms aimed at reasoning about the knowledge of autonomous agents about a given scenario.  Originally,  epistemic logics restricted their attention to so-called \emph{knowing that}, i.e.,  the capability of agents to know about certain facts.
More recently, several logics have been proposed to reason about alternative forms of knowledge (see~\cite{Wang16} for a discussion).  For instance, \emph{knowing whether} is looked into in~\cite{FWvD15}; \emph{knowing why} in~\cite{XuW16};  and \emph{knowing the value} in~\cite{GW16,Baltag16}, just to mention a few. Finally, a novel approach focuses on \emph{knowing how} --related to an agent's ability to achieve a goal~\cite{fantl2012introduction}. This concept is particularly interesting, as it has been argued to provide a fresh way to reason  about scenarios involving strategies in AI, such as those found in automated planning (see, e.g.,~\cite{HBEL}).

The first attempts to capture knowing how were through a combination of  `knowing that' and actions (see, e.g.,~\cite{Mccarthy69,Moore85,Les00,HerzigT06}).  However, it has been discussed, e.g., in~\cite{wiebeetal:2003,Herzig15}, that this idea does not lead to an accurate representation of knowing how. In response, a new logic is presented in~\cite{Wang15lori,Wang2016} featuring an original modality specifically tailored to model the concept of `knowing how'.  In a nutshell, an agent knows how to a achieve a goal $\varphi$ under some initial condition $\psi$, written $\kh(\psi,\varphi)$, if and only if there exists a `proper' plan $\plan$, i.e., a finite sequence of actions, that unerringly leads the agent from situations in which $\psi$ holds only to situations in which $\varphi$ holds. A `proper' plan is taken as one whose execution never aborts, an idea that takes inspiration from the notion of \emph{strong executability} from contingent planning~\cite{Smith&Weld98}.
As discussed in, e.g.,~\cite{JamrogaA07,Herzig15}, the quantification pattern we just described cannot be captured using logics with `knowing that' modalities and actions. 
For this reason, the new~$\kh$ modality from~\cite{Wang15lori,Wang2016} has reached a certain consensus in the community as an accurate way of modelling `knowing how'.
Moreover, it has paved the way to a deep study of knowing how, and to a rich family of logics capturing variants of the initial reading. Some examples of which are a ternary modality of knowing how with intermediate constraints~\cite{LiWang17}; a knowing how modality with weak plans~\cite{Li17}; a local modality for strategically knowing how~\cite{FervariHLW17} (and some relatives, see~\cite{Naumov2018a,NaumovT18}); and, finally, a knowing how modality which considers an epistemic indistinguishability relation among plans~\cite{AFSVQ21}.

As witnessed by all the ideas it triggered, the foundational work in~\cite{Wang15lori,Wang2016} greatly improved the understanding of `knowing how' from a logical standpoint.
The literature on logics of `knowing how' explores a wide variety of results, such as axiom systems (in most of the works cited above), proof methods~\cite{Li21,Li23}, and expressivity~\cite{FervariVQW22}, just to name a few. 
Yet, if we consider `knowing how' logics as suitable candidates for modelling problems in strategic reasoning, it is important to consider how difficult (or how easy) it is to use these logics for reasoning tasks.
There have been some recent developments on the complexity of logics with `knowing how' modalities.
For instance, model-checking for the $\kh$ modality above, and some of its variants, is investigated in~\cite{DF23}.
The complexity of model-checking and the decidability status of satisfiability for the local `knowing how' modality from~\cite{FervariHLW17}, and some of its generalizations, is explored in~\cite{LiW21}.
These two problems are also explored for `knowing how' with epistemic indistinguishability in~\cite{AFSVQ21}.
Notwithstanding, the complexity of the satisfiability problem for the original $\kh$ modality from~\cite{Wang15lori,Wang2016} is still unknown (\cite{Li17bis} 
presents only a decidability statement).
In this work, we shed some light into this matter.

Our contribution is to provide an upper for the satisfiability problem of the knowing how logic from~\cite{Wang15lori,Wang2016}, called here $\KHlogic$. 
More precisely, we introduce an algorithm for deciding satisfiability that is in $\Sigma^\Poly_2$, the second level of the polynomial hierarchy (\PH)~\cite{Stock76}. In short, this complexity class can be though as those problems invoking an $\NP$ oracle a polynomial number of times, and whose underlying problem is also in~\NP (see e.g.~\cite{AroraB09}).  Currently,  it is unknown whether \PH collapses, or it is strictly contained in \PSPACE. 
This being said, having an algorithm in a lower level of \PH is generally understood as a good indication that the problem is close to, e.g., \NP or~{\CoNP}. 
It is easy to see that \NP is a lower bound for checking satisfiability in $\KHlogic$, as it extends propositional logic. 
For an upper bound, a natural candidate is \PSPACE, as for instance the model-checking problem for $\KHlogic$ is \PSPACE-complete~\cite{DF23}, a potentially higher complexity of what is proved here for satisfiability.
We argue that this is due to the fact that in model-checking the full expressivity of the semantics is exploited (specially related to properties of regular languages), whereas for satisfiability, all this expressivity is completely hidden.
Although our procedure does not lead to a tight complexity characterization, it gives us an interesting upper bound towards filling this gap. 

We put forth that our result is not obvious. To obtain it, we combine techniques such as defining a normal form to eliminate nested modalities, calling an \NP oracle to guess propositional valuations and computing a closure over a matrix of formulas to combine them, adapting the Floyd-Warshall algorithm~\cite{Cormen22}. 


The article is organized as follows. In~\Cref{sec:basic} we introduce some notation and the basic definitions of the logic $\KHlogic$. \Cref{sec:upper} is devoted to incrementally show our result. Finally, in \Cref{sec:final} we provide some remarks and future lines of research.


\section{Knowing How Logic}
\label{sec:basic}
From here onwards,  we assume $\PROP$ is  a denumerable set of \emph{proposition symbols}, and $\ACT$ is a denumerable set of \emph{action symbols}. 
We refer to $\plan \in \PLANS$ as a \emph{plan}.

\begin{definition}\label{def:khsyntax}
    The \emph{language} $\KHlogic$ is determined by the grammar:
    \begin{align*}
        \varphi,\psi & \Coloneqq
            p \mid
            {\lnot \varphi} \mid
            {\varphi \lor \psi} \mid
            \kh(\varphi,\psi),
    \end{align*}
    where $p\in\PROP$. We use $\bot$, $\top$, $\varphi \land \psi$, $\varphi \to \psi$, and ${\varphi \liff \psi}$ as the usual abbreviations; $\A\varphi$ is defined as $\kh(\lnot\varphi,\bot)$ (see e.g.~\cite{Wang15lori,Wang2016}), while $\E\varphi$ abbreviates $\lnot\A\lnot\varphi$.
    The elements of $\KHlogic$ are \emph{formulas}.
\end{definition}

We read $\kh(\varphi,\psi)$ as: \emph{``the agent knows how to achieve $\psi$ given $\varphi$''}.
We call~$\varphi$ and $\psi$, the precondition and the postcondition of $\kh(\varphi,\psi)$, respectively. We read $\A\varphi$ as: \emph{``$\varphi$ holds anywhere''}; and its dual $\E\varphi$ as: \emph{``$\varphi$ holds somewhere''}.
As it is usually done, we refer to $\A$ and $\E$ as the \emph{universal} and \emph{existential} modalities~\cite{GorankoP92}. 

Formulas of $\KHlogic$ are interpreted with respect to \emph{labelled transition systems} over so-called \emph{strongly executable plans}.
Sometimes, we refer to LTS as \emph{models}. 
We introduce their definitions below.

\begin{definition}\label{def:lts}
    A \emph{labelled transition system (\lts)} is a tuple $\model = \tup{\S,\R,\V}$ s.t.:
    \begin{enumerate}[(1)]
    \item $\S$ is a non-empty set of \emph{states};
    \item ${\R = \setof{\R_a}{a \in \ACT}}$ is a collection of binary relations on $\S$; and
    \item $\V: \PROP \to 2^{\S}$ is a \emph{valuation function}.
    \end{enumerate}
\end{definition}

\begin{definition}\label{def:plans-aux}
    Let $\csetsc{\R_a}{a \in\ACT}$ be a collection of binary relations on $\S$.
    Let $\varepsilon\in\ACT^*$ be the empty plan.
    We define:  $\R_{\varepsilon} = \setof{(s,s)}{s \in \S}$, and 
    for every $\plan \in \PLANS$, and $a \in \ACT$,
        $\R_{\plan{a}} = \R_{\plan}\R_{a}$ (i.e., their composition). For every relation $\R_\plan$, and $T\subseteq\S$, define $\R_\plan(T)=\setof{(s,t)}{s \in T \mbox{ and } (s,t)\in\R_\plan}$, and $\R_\plan(t)=\R_\plan(\set{t})$.
\end{definition}

The notion of \emph{strong executability} determines the ``adequacy'' of a plan.
Strong executability takes inspiration from conformant planning~\cite{Smith&Weld98}, and its jusification is discussed at length in~\cite{Wang15lori}.

\begin{definition} \label{def:plans-exec}
    Let $\plan=a_1 \dots a_n \in \ACT^*$, and $1 \leq i \leq j \leq n$, we denote: 
    $\plan_i = a_i$;
    $\plan[i,j] = a_i \dots a_j$; and
    $|\plan| = n$.
    Moreover, let $\model = \tup{\S, \R, \V}$ be an LTS;
    we say that $\plan$ is \emph{strongly executable (SE)} at
    $s \in \S$, iff for all $i \in [1,\card{\plan}-1]$ and all $t \in \R_{(\plan[1,i])}(s)$, it follows that $\R_{\plan_{(i+1)}}(t) \neq \emptyset$. 
    The set of all states at which $\plan$ is strongly executable is defined as $\stexec(\plan) = \setof{s}{\plan \ \mbox{\rm is SE at } s}$.
    Note: $\stexec(\varepsilon) = \S$. 
\end{definition}

We illustrate the notions we just introduced with a simple example.
\begin{example}\label{ex:rel-notions}
    Let $\model=\tup{\S,\R,\V}$ be the LTS depicted below and $\plan=ab$. We have, $\R_{\plan}(s)=\set{u}$, and  $\R_{\plan[1,1]}(s)=\R_{a}(s)=\set{t,v}$. It can be seen that $s \in \stexec(a)$; while $s \notin \stexec(\plan)$ --since $v\in\R_{\plan[1,1]}(s)$ and $\R_{\plan_{(2)}}(v)=\R_{b}(v)=\emptyset$. Finally, we have that $\stexec(\varepsilon)=\S$, $\stexec(a)=\set{s}$ and $\stexec(ab)=\emptyset$.

\begin{center}
    \begin{tikzpicture}[->]
        \node [state,label=above:\small$s$] (w1) {$p$};
        \node [state, label=above:\small$t$,right = of w1] (w2) {$r$};
        \node [state, label=above:\small$v$, below = of w2] (w3) {$r$};
        \node [state, label=above:\small$u$,right = of w2] (w4) {$q$};


        \path (w1) edge node [label-edge, above] {$a$} (w2)
            (w2) edge node [label-edge, above] {$b$} (w4)
            (w1) edge node [label-edge, above] {$a$} (w3);

    \end{tikzpicture}
\end{center}
\end{example}

We are now ready to introduce the semantics of $\KHlogic$, based on~\cite{Wang15lori,Wang2016}.

\begin{definition} \label{def:khsemantics}
    Let $\model = \tup{\S,\R,\V}$ be an \lts, we define $\truthset{\model}{\varphi}$ inductively as:
    \begin{align*}
        \truthset{\model}{p} &=
            \V(p) \quad \quad 
        \truthset{\model}{\lnot \varphi} =
            \S \setminus \truthset{\model}{\varphi}  \quad \quad 
        \truthset{\model}{\varphi \lor \psi} =
            \truthset{\model}{\varphi} \cup \truthset{\model}{\varphi} \\
        \truthset{\model}{\kh(\varphi,\psi)} &=
            \begin{cases}
                \S & \text{if exists } \plan{\in}\PLANS \text{s.t. } 
                        \truthset{\model}{\varphi} \subseteq \stexec(\plan) \text{ and }
                        \R_{\plan}(\truthset{\model}{\varphi}) \subseteq \truthset{\model}{\psi} \\
                \emptyset & \text{otherwise}. \\
            \end{cases}
    \end{align*}
   We say that a plan $\plan \in \PLANS$ is a \emph{witness} for $\kh(\varphi,\psi)$ iff 
       $\truthset{\model}{\varphi} \subseteq \stexec(\plan)$ and
       $\R_{\plan}(\truthset{\model}{\varphi}) \subseteq \truthset{\model}{\psi}$.
    We use $(\truthset{\model}{\varphi})^{\complement}$ instead of $\S \setminus \truthset{\model}{\varphi}$.
    We write $\model \Vdash \varphi$ as an alternative to $\truthset{\model}{\varphi} = \S$; and $\model, s \Vdash \varphi$ as an alternative to $s \in \truthset{\model}{\varphi}$.
\end{definition}

\begin{example}\label{exs:semantics}
Let $\model$ be the LTS from~\Cref{ex:rel-notions}. From~\Cref{def:khsemantics}, we have ${\truthset{\model}{\kh(p,r)}=\S}$ (using $a$ as a witness), while $\truthset{\model}{\kh(p,q)}=\emptyset$ (there is no witness for the formula).
\end{example}

We included the universal modality $\A$ as abbreviation since formulas of the form $\A\varphi$ play a special role in our treatment of the complexity of the satisfiability problem for $\KHlogic$. It is proven in, e.g.,~\cite{Wang15lori,Wang2016}, that $\A\varphi$ and $\E\varphi$ behave as the universal and existential modalities (\cite{GorankoP92}), respectively. Recall that $\A\varphi$ is defined as $\kh(\neg\varphi,\bot)$, which semantically states that $\varphi$ holds everywhere in a model iff $\neg\varphi$ leads always to impossible situations. Formulas of this kind are called here `global'. Below, we formally restate the results just discussed.

\begin{proposition}\label{prop:ik:universal}
     Let $\model = \tup{\S,\R,\V}$ and $\psi$ and $\chi$ be formulas s.t.\ $\truthset{\model}{\chi}=\emptyset$; then $\truthset{\model}{\kh(\psi,\chi)}=\S$ iff $\truthset{\model}{\lnot\psi}=\S$.
\end{proposition}

\begin{corollary}\label{prop:nts:universal}
    Let $\model = \tup{\S,\R,\V}$ and a formula $\varphi$, 
   $\model,s\Vdash \A\varphi$ iff $\truthset{\model}{\varphi}=\S$.
\end{corollary}

We introduce now \Cref{prop:pre:post}, which is of use in the rest of the paper.

\begin{restatable}{proposition}{propcomposition}\label{prop:pre:post}\label{prop:composition}
   Let $\psi,\psi',\chi,\chi'$ and $\varphi$ be formulas, and $\model$ an LTS; then:
    \begin{enumerate}[(1)]
        \item
            $\truthset{\model}{\psi'} \subseteq \truthset{\model}{\psi}$ and
            $\truthset{\model}{\chi} \subseteq \truthset{\model}{\chi'}$ implies
                $\truthset{\model}{\kh(\psi,\chi)} \subseteq \truthset{\model}{\kh(\psi',\chi')}$;
        \item
            $\truthset{\model}{\psi} \subseteq \truthset{\model}{\psi'}$ implies
                $(\truthset{\model}{\kh(\varphi,\psi)} \cap \truthset{\model}{\kh(\psi',\chi)}) \subseteq \truthset{\model}{\kh(\varphi,\chi)}$.
    \end{enumerate}
\end{restatable}

We conclude this section with some useful definitions. 

\begin{definition}\label{def:satisfiable}
    A formula $\varphi$ is \emph{satisfiable}, written $\sat(\varphi)$, iff there is $\model$ s.t.\ $\truthset{\model}{\varphi} \neq \emptyset$.
    A finite set $\Phi = \{\varphi_1, \dots, \varphi_n \}$ of formulas is \emph{satisfiable}, written $\sat(\Phi)$, iff $\sat(\varphi_1 \land \dots \land \varphi_n)$.
    For convenience, we define $\sat(\emptyset)$ as true.
    We use $\unsat(\varphi)$ iff $\sat(\varphi)$ is false; similarly,  $\unsat(\Phi)$ iff $\sat(\Phi)$ is false.
    Finally, whenever $\sat(\varphi)$ iff $\sat(\varphi')$, we call $\varphi$ and $\varphi'$ \emph{equisatisfiable}, and write $\varphi \equiv_{\sat} \varphi'$.
\end{definition}

\begin{definition}\label{def:md}
    The modal depth of a formula $\varphi$, written $\depth(\varphi)$, is defined as:
    \begin{align*}
        \depth(\varphi) & =
            \begin{cases}
                0 & \text{if } \varphi \in \PROP \\
                \depth(\psi) & \text{if } \varphi = \lnot \psi \\
                \max(\depth(\psi),\depth(\chi)) & \text{if } \varphi = \psi \lor \chi \\
                1+\max(\depth(\psi),\depth(\chi)) & \text{if } \varphi = \kh(\psi,\chi).
            \end{cases}
    \end{align*}
    We use $\sforms(\varphi)$ to indicate the set of subformulas of $\varphi$.
    We say that $\kh(\psi,\chi)$ is a leaf of $\varphi$ iff $\kh(\psi,\chi) \in \sforms(\varphi)$ and $\depth(\psi) = \depth(\chi) = 0$ (i.e., $\depth(\kh(\psi,\chi)=1)$).
\end{definition}
  
In words, the modal depth of a formula is the length of the longest sequence of nested modalities in the formula; whereas a leaf is a subformula of depth one. Notice that, since $\A\varphi$ is a shortcut for $\kh(\neg\varphi,\bot)$, we have $\depth(\A\varphi)=1+\depth(\varphi)$. 

\begin{example}\label{ex:md}
    Let $\varphi = \kh(p, \kh(\lnot q, p \to q)) \lor \kh(r,t)$;
     it can easily be checked that $\depth(\varphi) = 2$ and that $\kh(\lnot q, p \to q)$ and $\kh(r,t)$ are its modal leaves.
\end{example}

\section{An Upper Bound for the Satisfiability Problem of $\KHlogic$}
\label{sec:upper}
In this section we establish an upper bound on the complexity of the satisfiability problem for $\KHlogic$, which is the main result of our paper. 
We start with some preliminary definitions and results. 

\begin{restatable}{proposition}{propflatten}\label{prop:flatten}\label{prop:composition}
Let $\varphi'$ be the result of replacing all occurrences of a leaf $\theta$ in $\varphi$ by a proposition symbol $k \notin \sforms(\varphi)$; it follows that $\varphi \equiv_{\sat} (\varphi' \land (\A{k} \lra \theta))$.
\end{restatable}

We say that $\varphi$ is in \emph{leaf normal form} iff $\depth(\varphi) \leq 1$.
\Cref{prop:form2form1} tells us that we can put any formula into an equisatisfiable formula in leaf normal form.
The function {\sc Flatten} in \Cref{alg:form2form1} tells us how to do this in polynomial time.

\begin{proposition}\label{prop:form2form1}
  \Cref{alg:form2form1} is in $\Poly$; on input $\varphi$, it outputs $\varphi_0$ and $\varphi_1$ such that $\depth(\varphi_0) = 0$, $\depth(\varphi_1) = 1$, and $\varphi \equiv_{\sat} (\varphi_0 \land \varphi_1)$.
\end{proposition}

\begin{algorithm}[!t]
  \caption{Flatten}
  \label{alg:form2form1}
  \begin{algorithmic}[1]
    \Require true
    \Function{Flatten}{$\varphi$}
      \State $\varphi_0, \varphi_1 \gets \varphi, \top$
      \Loop \Comment{\textbf{invariant:} $\varphi \equiv_{\sat} (\varphi_0 \land \varphi_1)$ (see \Cref{prop:flatten})}
        \State $\Theta \gets \text{the set of leaves of $\varphi_0$}$
        \IIf {$\Theta = \emptyset$}
          \textbf{break}
        \EndIIf \Comment{loop guard}
        \ForAll{ $\theta \in \Theta$ }
          \State $k\hspace*{3.7pt} \gets \text{a proposition symbol not in } \sforms(\varphi_0 \land \varphi_1)$
          \State $\varphi_0 \gets \text{result of replacing all occurrences of $\theta$ in $\varphi_0$ for $k$}$
          \State $\varphi_1 \gets \varphi_1 \land (\A{k} \liff \theta)$
        \EndFor
      \EndLoop
      \State \Return $\varphi_0 \land \varphi_1$
    \EndFunction
    \Ensure
      $\depth(\varphi_0) = 0$ and
      $\depth(\varphi_1) = 1$ and
      $\varphi \equiv_{\sat} (\varphi_0 \land \varphi_1)$
  \end{algorithmic}
\end{algorithm}

The result in \Cref{prop:form2form1} allows us to think of the complexity of the satisfiability problem for $\KHlogic$ by restricting our attention to formulas in leaf normal form.
In turn, this enables us to work towards a solution in terms of subproblems.
More precisely, given $\varphi_0$ and $\varphi_1$ in the leaf normal form that results from {\sc Flatten}, the subproblems are 
    (i) determining the satisfiability of $\varphi_0$; and
    (ii) determining the satisfiability of $\varphi_1$ based on a solution to (i).
The solution to (i) is well-known, $\varphi_0$ is a propositional formula.
We split the solution of (ii) into 
    (a) determining when formulas of the form $\kh(\psi_1,\chi_1) \land \dots \land \kh(\psi_n,\chi_n)$ are satisfiable, see~\Cref{prop:sat+};
    (b) determining when formulas of the form $\lnot\kh(\psi'_1,\chi'_1) \land \dots \land \neg\kh(\psi'_m,\chi'_m)$ are satisfiable, see~\Cref{prop:sat-}; and
    (c) combining (a) and (b), see~\Cref{prop:compatible}.
We present (a), (b), and (c), in a way such that they incrementally lead to a solution to the satisfiability problem for $\KHlogic$.
Finally, in~\Cref{prop:alg:compatible}, we show how to combine (i) and (ii) to obtain an upper bound on the complexity of this problem. 

Let us start by solving the first problem: checking whether a conjunction $\varphi$ of positive formulas in leaf normal form are satisfiable altogether. In a nutshell, we show that solving this problem boils down to building a set $I$ of the preconditions of those subformulas whose postconditions are falsified in the context of $\varphi$, and checking whether the formulas in $I$ are satisfiable or not.
Intuitively, the formulas in $I$ correspond to `global' formulas. 
We made precise these ideas in~\Cref{prop:sat+}.

\begin{proposition} \label{prop:sat+}
    Let $\varphi = \kh(\psi_1,\chi_1) \land \dots \land \kh(\psi_n,\chi_n)$ be such that $\depth(\varphi) = 1$; and let 
    the sets $I_0,\dots,I_n$ be defined as follows:       
        \[
            I_i =
                \begin{cases}
                    \setof{k \in [1,n]}{\unsat(\chi_k)} & \text{if } i=0, \\
                    I_{(i-1)} \cup \setof{k \in [1,n]}{\unsat(\setof{\lnot \psi_{k'}}{k' \in I_{(i-1)}} \cup \set{\chi_k})} & \text{if } i>0,
                \end{cases}          
        \]
        where $i \in [0,n]$; further, let $I=I_n$.
    Then: (1) $\sat(\varphi)$ iff (2) $\sat(\bigwedge_{i \in I}\lnot \psi_i)$.
\end{proposition}
\begin{proof}
    ($\Rightarrow$)
    Suppose that $\sat(\varphi)$ holds, i.e., exists $\model$ s.t.\ $\truthset{\model}{\varphi} = \S$.
    From this assumption, we know that, for every $j \in [1,n]$, $\truthset{\model}{\kh(\psi_i,\chi_i)}=\S$.
    The proof is concluded if $\truthset{\model}{\bigwedge_{i \in I}\lnot \psi_i} \neq \emptyset$.
    We obtain this last result with the help of the following auxiliary lemma:
    \[
        \text{($*$) for all $k \in I_i$, $\truthset{\model}{\chi_k} = \emptyset$ and $\truthset{\model}{\lnot \psi_k} = \S$}
    \]
    The lemma is obtained by induction on the construction of $I_i$.
    The base case is direct.
    Let ${k \in I_0}$; from the definition of $I_0$, we get $\unsat(\chi_k)$; this implies ${\truthset{\model}{\chi_k} = \emptyset}$; which implies $\S = \truthset{\model}{\kh(\psi_k,\chi_k)} = \truthset{\model}{\A \lnot \psi_k} = \truthset{\model}{\lnot \psi_k}$.
    For the inductive step, let $k \in I_{(i+1)} \setminus I_i$.
    From the Inductive Hypothesis, for all ${k' \in I_i}$, $\truthset{\model}{\chi_{k'}} = \emptyset$ and $\truthset{\model}{\lnot \psi_{k'}} = \S$.
    This implies (\dag) $\truthset{\model}{\bigwedge_{k' \in I_{i}} \lnot \psi_{k'}} = \S$.
    From the definition of $I_{(i+1)}$,
        ${\unsat(\setof{\lnot \psi_{k'}}{k' \in I_{i}} \cup \set{\chi_k})}$.
    This is equivalent to $\truthset{\model}{\bigwedge_{k' \in I_{i}} \lnot \psi_{k'}} \subseteq \truthset{\model}{\lnot \chi_k}$.
    From (\dag), $\S \subseteq \truthset{\model}{\lnot \chi_k} = \S$.
    Thus, $\truthset{\model}{\chi_k} = \emptyset$ and $\truthset{\model}{\lnot \psi_k} = \S$.
    Since $I = I_n$; using ($*$) we get $\truthset{\model}{\bigwedge_{i \in I}\lnot \psi_i} = \S \neq \emptyset$.
    This proves~(2).

    \medskip
    \noindent
    ($\Leftarrow$) The proof is by contradiction.
    Suppose (2) and $\unsat(\varphi)$.
    Then, for all $\model$, ($\dagger$) $\truthset{\model}{\varphi} = \emptyset$.
    Let $J = \setof{j \in [1,n]}{\unsat(\set{(\bigwedge_{i \in I}\lnot \psi_i),\psi_j})}$.
    Moreover, let $\model = \tup{\S,\R,\V}$ be s.t.\ $\S$ is the smallest set containing all valuations that make $(\bigwedge_{i \in I}\lnot \psi_i)$ true.
    From (2), we know that $\S \neq \emptyset$ and $\truthset{\model}{\lnot \psi_k} = \S$ for all $k \in I$.
    By induction on the construction of $I=I_n$, we get that $\truthset{\model}{\chi_k} = \emptyset$ for all $k \in I = \bigcup_{i=0}^n I_i$.
    The case for $k \in I_0$ is direct since $\unsat(\chi_k)$, thus $\truthset{\model}{\chi_k} = \emptyset$. For the inductive case, let $k \in I_i \setminus I_{i-1}$, then $\unsat(\setof{\lnot \psi_{k'}}{k' \in I_{(i-1)}} \cup \set{\chi_k})$.
    This is equivalent to say that the implication $((\bigwedge_{k' \in I_{(i-1)}} \lnot \psi_{k'}) \rightarrow \lnot \chi_k)$ is valid. Thus, $\truthset{\model}{\bigwedge_{k' \in I_{(i-1)}} \lnot \psi_{k'}} \subseteq \truthset{\model}{\lnot \chi_k}$.
    By hypothesis, $\truthset{\model}{  \bigwedge_{k' \in I} \lnot \psi_{k'}} = \S$. Thus, $\truthset{\model}{  \bigwedge_{k' \in I_{(i-1)}} \lnot \psi_{k'}} = \S$, and we get $\truthset{\model}{\lnot \chi_k} = \S$ and $\truthset{\model}{\chi_k} = \emptyset$.
    In turn, for all $k \in J$, since $\unsat(\set{(\bigwedge_{i \in I}\lnot \psi_i),\psi_k})$ and $\truthset{\model}{\bigwedge_{i \in I}\lnot \psi_i} = \S$ we can conclude that $\truthset{\model}{\psi_k} = \emptyset$.
    Thus, we have that
        $\truthset{\model}{\kh(\psi_k,\chi_k)} = \truthset{\model}{\A \lnot \psi_k} = \S$, 
            for all $k \in {I \cup J}$.
    Then, from ($\dagger$), exists $K = \setof{k}{\truthset{\model}{\kh(\psi_k,\chi_k)} = \emptyset}$ s.t.\ $\emptyset \subset K \subseteq [1,n]\setminus (I \cup J)$.
    For all $k \in K$, $\truthset{\model}{\psi_k} \neq \emptyset$ since $\sat(\set{(\bigwedge_{i \in I}\lnot \psi_i),\psi_k})$; and $\truthset{\model}{\chi_k} \neq \emptyset$ since $\sat(\setof{\lnot \psi_{k'}}{k' \in I_{(i-1)}} \cup \set{\chi_k})$ for all $i \geq 0$, even $I_{(i-1)} = I_n = I$.
    Without loss of generality, let $K = [1,m]$ and $\model' = \tup{\S,\R',\V}$ be s.t. $\R' = \setof{\R'_{a_j}}{a_j \in \ACT}$, where:
    \begin{align*}
        \R'_{a_j} &=
            \begin{cases}
                {\truthset{\model'}{\psi_j} \times \truthset{\model'}{\chi_j}}
                    & \text{if } j \in K \\
                \R_{a_{(j-m)}}
                    & \text{if } j \notin K.
            \end{cases}
    \end{align*}
    In the definition of $\R'$, it is worth noticing that since $j\notin K$, $\R_{a_{(j-m)}}$ is defined, i.e., $\R_{a_{(j-m)}} \in \R$. 
    Then clearly, for all $k \in K$, $\truthset{\model'}{\kh(\psi_k,\chi_k)} = \S$.
    The claim is that 
        for all $k' \in I \cup J$,
            $\truthset{\model'}{\kh(\psi_{k'},\chi_{k'})} = \S$.
    To prove this claim, consider a function
        $\sigma: \ACT^{*} \to \ACT^{*}$ s.t. 
            $\sigma(\varepsilon) = \varepsilon$, and
            $\sigma(a_k\alpha) = a_{(k+m)}\sigma(\alpha)$.
    For all $\pi \in \PLANS$, if $\truthset{\model}{\psi_{k'}} \subseteq \stexec(\pi)$ and $\R_{\pi}(\truthset{\model}{\psi_{k'}}) \subseteq \truthset{\model}{\chi_{k'}}$, then $\truthset{\model'}{\psi_{k'}} \subseteq \stexec(\sigma(\pi))$ and $\R_{\sigma(\pi)}(\truthset{\model'}{\psi_{k'}}) \subseteq \truthset{\model'}{\chi_{k'}}$ --since the valuation functions for $\model$ and $\model'$ coincide, the truth sets in $\model$ and $\model'$ coincide for formulas with no modalities. 
    Then, $\truthset{\model'}{\kh(\psi_{k'},\chi_{k'})} = \S$.
    But we had assumed $\unsat(\varphi)$.
    Thus, (1) follows.
\end{proof}

The following example illustrates the result in \Cref{prop:sat+}.

\begin{example}\label{ex:global}
Let $\varphi= \kh(p,\bot) \wedge \kh(q,p)$, i.e., $\psi_1=p$, $\psi_2=q$, $\chi_1=\bot$ and $\chi_2=p$. It is clear that $\sat(\varphi)$. Let us build the sets $I_0$, $I_1$ and $I_2$:
\begin{itemize}
    \item $I_0= \set{1}$, as $\unsat(\chi_1)$ and $\sat(\chi_2)$ hold;
    \item $I_1 = \set{1,2}$, since it holds $\unsat(\set{\neg\psi_1,\chi_2}$);
    \item $I_2 = \set{1,2} = I$, as $I_1$ already contains all the indices in $[1,2]$.
\end{itemize}
Thus (as it can be easily checked) we get $\sat(\set{\neg\psi_1,\neg\psi_2})$ (i.e., $\sat(\set{\neg p, \neg q})$).
\end{example}

Interestingly, the result in \Cref{prop:sat+} tells us that the satisfiability of a formula $\kh(\psi_1,\chi_1) \land \dots \land \kh(\psi_n,\chi_n)$ depends solely on the joint satisfiability of its `global' subformulas (cf.~\Cref{prop:ik:universal}); i.e., subformulas $\kh(\psi_i,\chi_i)$ whose postconditions $\chi_i$ are falsified in the context of $\varphi$. 
The satisfiability of the `global' subformulas provides us with the universe, i.e., set of states, on which to build the plans that witness those formulas that are not in $I$, and that are not `trivially' true as a result of their preconditions being falsified in this universe.

Building on~\Cref{prop:sat+}, the function {\sc Sat$^{+}_{\kh}$} in \Cref{alg:sat+} gives us a way of checking whether a formula $\varphi = \kh(\psi_1,\chi_1) \land \dots \land \kh(\psi_n,\chi_n)$ is satisfiable.
The algorithm behind this function makes use of a (propositional) $\sat$ oracle, and the function {\sc Global}.
The $\sat$ oracle tests for pre and postconditions of $\kh$ formulas, as these are propositional formulas.
Intuitively, {\sc Global} iteratively computes the indices in the sets $I_i$ in \Cref{prop:sat+}, each of them corresponding to the `global' subformulas of the input.
Once this is done, {\sc Sat$^{+}_{\kh}$} checks the joint satisfiability of the negation of the preconditions of `global' subformulas.

\begin{proposition}\label{prop:alg:sat+}
    Let $\varphi$ be as in \Cref{prop:sat+}; \Cref{alg:sat+} solves $\sat(\varphi)$.
\end{proposition}
  


\begin{algorithm}[!t]
    \caption{{\sc Sat}$^{+}_{\kh}$}
    \label{alg:sat+}
    \begin{algorithmic}[1]
    \Require
        $\varphi = \kh(\psi_1,\chi_1) \land \dots \land \kh(\psi_n,\chi_n)$ and
        $\depth(\varphi) = 1$
    \Function{Global}{$\varphi$}
        \State $I,\Psi \gets \emptyset,\emptyset$
        \For{$i \gets 0$ {\bf to} $n$}
            \State $K \gets \emptyset$
            \For{$k \gets 1$ {\bf to} $n$}
                \IIf {not $\sat(\Psi \cup \{\chi_k\})$}
                    $K \gets K \cup \{k\}$
                \EndIIf
            \EndFor
            \State $I \gets {I \cup K}$
            \State $\Psi \gets \Psi \cup \setof{\lnot \psi_{k}}{k \in K}$
        \EndFor
        \State \Return $I$
    \EndFunction
    \Ensure
    $\Call{Global}{\varphi} = {I_0 \cup \dots \cup I_n}$ where $I_{i}$ is as in \Cref{prop:sat+}
    \Require
        $\varphi = \kh(\psi_1,\chi_1) \land \dots \land \kh(\psi_n,\chi_n)$ and
        $\depth(\varphi) = 1$
    \Function{Sat$^{+}_{\kh}$}{$\varphi$}
        \State \Return $\sat(\setof{\lnot \psi_i}{i \in \Call{Global}{\varphi}})$
    \EndFunction
    \Ensure
        \Call{Sat$^{+}_{\kh}$}{$\varphi$} iff $\sat(\varphi)$
    \end{algorithmic}
\end{algorithm}

Let us now move to determining the satisfiability conditions of a formula $\lnot \kh(\psi_1,\chi_1) \land \dots \land \lnot \kh(\psi_n,\chi_n)$ in leaf normal form.
\Cref{prop:sat-} establishes that, for any such a formula, it is enough to check whether each conjunct $\psi_i\land\lnot\chi_i$ is individually satisfiable.
Note that this satisfiability check is purely propositional.

\begin{proposition} \label{prop:sat-}
    Let $\varphi = \lnot \kh(\psi_1,\chi_1) \land \dots \land \lnot \kh(\psi_n,\chi_n)$ be s.t.\ $\depth(\varphi) = 1$;
    it follows that
        $\sat(\varphi)$ iff
        for all $i \in [1, n]$, $\sat(\psi_i \land \lnot \chi_i)$.
\end{proposition}
\begin{proof}
    ($\Rightarrow$)
    The proof is by contradiction.
    Suppose that ($\dagger$) $\sat(\varphi)$ and for some  $i \in [1,n]$ we have ($\ddagger$) $\unsat(\psi_i \land \lnot \chi_i)$. Let $\model$ be a model such that
    $\truthset{\model}{\varphi} \neq \emptyset$, which exists by ($\dagger$).
    Then, $\truthset{\model}{\kh(\psi_i,\chi_i)} = \emptyset$.
    From this, we get $\truthset{\model}{\psi_i} \neq \emptyset$; otherwise $\truthset{\model}{\kh(\psi_i,\chi_i)} = \S$.
    From ($\ddagger$), we know that $\truthset{\model}{\psi_i} \subseteq \truthset{\model}{\chi_i}$.
    Since $\varepsilon \in \PLANS$, we have  $\truthset{\model}{\psi_i} \subseteq \stexec(\varepsilon) = \S$ and $\truthset{\model}{\psi_i} = \R_{\varepsilon}(\truthset{\model}{\psi_i}) \subseteq \truthset{\model}{\chi_i}$.
    But this means $\truthset{\model}{\kh(\psi_i,\chi_i)} = \S$; which is a contradiction.
    Thus, $\R_{\varepsilon}\truthset{\model}{\psi_i} \nsubseteq \truthset{\model}{\chi_i}$; i.e., $\truthset{\model}{\psi_i} \nsubseteq \truthset{\model}{\chi_i}$. This means $\truthset{\model}{\psi_i \land \lnot \chi_i} \neq \emptyset$.
    This establishes $\sat(\psi_i \land \lnot \chi_i)$.

    \medskip
    \noindent
    ($\Leftarrow$)
    Suppose that ($\dagger$) for all $i \in [1,n]$, $\sat(\psi_i \wedge \lnot\chi_i)$.
    Let $\model = \tup{\S, \R, \V}$ where:
        ($\ddagger$) $\S$ is s.t.\ for all $i$, $\tset{\psi_i \land \lnot \chi_i}^{\model} \neq \emptyset$; and
        ($\mathsection$) for all $\R_{a} \in \R$, $\R_{a} = \emptyset$.
    From ($\dagger$), we know that at least one $\S$ exists, as every $\psi_i$ and $\chi_i$ are propositional; thus, each satisfiable conjunction can be sent to a different $s \in\S$.
    From ($\mathsection$), we know for all $\pi \in \PLANS$, $\stexec(\pi) \neq \emptyset$ iff $\pi = \varepsilon$.
    From ($\ddagger$) and ($\mathsection$), we know that $\truthset{\model}{\psi_i} = \R_{\varepsilon}\tset{\psi_i}^{\model} \nsubseteq \tset{\chi_i}^{\model}$.
    This means that $\tset{\kh(\psi_i,\chi_i)}^{\model} = \emptyset$, for all $i \in [1,n]$.
    Hence $\tset{\varphi}^{\model} = \S$ which implies $\sat(\varphi)$.
\end{proof}

The key idea behind \Cref{prop:sat-} is to build a discrete universe to force the only possible witness of a formula of the form $\kh(\psi_i,\chi_i)$ to be the empty plan.
If in this discrete universe we always have at hand a state which satisfies $\psi_i \land \lnot \chi_i$, then, the empty plan cannot be a witness for $\kh(\psi_i,\chi_i)$.
If the latter is the case, then the satisfiability of $\lnot\kh(\psi_i,\chi_i)$ is ensured.
Building on this result, we define, in \Cref{alg:sat-}, a function {\sc Sat$^{-}_{\kh}$} to check the satisfiability of a formula $\lnot \kh(\psi_1,\chi_1) \land \dots \land \lnot \kh(\psi_n,\chi_n)$ in leaf normal form.
The function proceeds by traversing each subformula $\kh(\psi_i,\chi_i)$ and checking the satisfiability of $\psi_i \land \lnot\chi_i$.

\begin{proposition}\label{prop:alg:sat-}
    Let $\varphi$ be as in \Cref{prop:sat-}; \Cref{alg:sat-} solves $\sat(\varphi)$.
\end{proposition}
  


\begin{algorithm}[!t]
    \caption{{\sc Sat}$^{-}_{\kh}$}
    \label{alg:sat-}
    \begin{algorithmic}[1]
    \Require
        $\varphi = \lnot\kh(\psi_1,\chi_1) \land \dots \land \lnot\kh(\psi_n,\chi_n)$ and
        $\depth(\varphi) = 1$
    \Function{Sat$^{-}_{\kh}$}{$\varphi$}
        \State $r \gets \top$
        \For{$i \gets 1$ {\bf to} $n$}
            \State $r \gets r$ and $\sat(\psi_i \land \lnot \chi_i)$
        \EndFor
        \State \Return $r$
    \EndFunction
    \Ensure
        \Call{Sat$^{-}_{\kh}$}{$\varphi$} iff $\sat(\varphi)$
    \end{algorithmic}
\end{algorithm}

We are now ready to extend the results in \Cref{prop:sat+,prop:sat-} to work out the joint satisfiability of a formula of the form $\varphi^{+} = \kh(\psi_1,\chi_1) \land \dots \land \kh(\psi_n,\chi_n)$, and a formula of the form $\varphi^{-} = \neg\kh(\psi'_1,\chi'_1) \land \dots \land \lnot \kh(\psi'_m,\chi'_m)$, both in leaf normal form.
The main difficulty is how to ``build'' witnesses for the subformulas $\kh(\psi_i,\chi_i)$ of $\varphi^{+}$ in a way such that they do not yield witnesses for the subformulas $\lnot\kh(\psi'_j,\chi'_j)$ of $\varphi^{-}$.
We show that the key to the solution hinges on ``composition''. 
We start with a preliminary definition.

\begin{definition} \label{def:composition}
    Let
        $\varphi = \kh(\psi_1,\chi_1) \land \dots \land \kh(\psi_n,\chi_n)$ and
        $\psi$ be a formula; we define 
        $\C(\varphi,\psi) = \bigcup_{i \geq 0}\C_i$ where:
        \begin{align*}
            \C_0 &=
                \setof{(x,x)}{x \in [1,n]} \\
            \C_{(i+1)} &=
                \C_i \cup
                \setof{(x,z)}{
                    (x,y) \in \C_i
                    \text{, }
                    z \in [1,n]
                    \text{, and }
                    \unsat(\set{\psi, \chi_{y}, \lnot \psi_z})
                }.
        \end{align*}
\end{definition}

In words, $\C(\varphi,\psi)$ captures the notion of composition of formulas $\kh(\psi,\chi)$ and $\kh(\psi',\chi')$ into a formula $\kh(\psi,\chi')$.
This composition is best explained by recalling the validity of $(\kh(\psi,\chi)\land\A(\chi \to \psi')\land\kh(\psi',\chi')) \to \kh(\psi,\chi')$ (see, e.g.~\cite{Wang15lori,Wang2016}). 
The definition of $\C(\varphi,\psi)$ records the conjuncts of~$\varphi$ which can be composed in this sense.
Below, we list some properties of $\C(\varphi,\psi)$.

\begin{proposition} \label{prop:chain}
    Let $\varphi$ and $\psi$ be as in \Cref{def:composition};
    if
        $(x,y) \in \C(\varphi,\psi)$,
    then,
        for any model $\model$, it holds
         $\truthset{\model}{\varphi \land \A\psi} \subseteq \truthset{\model}{\kh(\psi_x,\chi_y)}$.
\end{proposition}
\begin{proof}
    We start by stating and proving an auxiliary lemma:
        ($*$) $(x,y) \in \C_i$ iff
        there is a non-empty sequence $\plan$ of indices in $[1,n]$ s.t.:
        \begin{enumerate}
            \item[(\dag)] $x=\plan_1$ and $y=\plan_{|\plan|}$; and
            \item[(\ddag)] 
            for all $j \in [1,|\plan|-1]$,
                $\unsat(\set{\psi, \chi_{\plan_j}, \lnot \psi_{\plan_{(j+1)}}})$.
        \end{enumerate}
    The proof of this lemma is by induction on $i$.
    The base case for ($*$) is $i=0$.
    We know that $(x,x) \in \C_0$, the sequence containing just $x$ satisfies (\dag) and (\ddag).
    Conversely, we know that any sequence $\plan$ of indices in $[1,n]$ s.t.\ $|\plan| = 1$ satisfies (\dag) and (\ddag); it is immediate that $(\plan_1,\plan_1) \in \C_0$.
    This proves the base case.
    For the inductive step, let
        $(x,z) \in \C_{(i+1)}$,
        $(x,y) \in \C_i$,
        $z \in [1,n]$, and
        $\unsat(\set{\psi, \chi_y, \lnot \psi_z})$.
    From the Inductive Hypothesis, there is $\plan$ that satisfies (\dag) and (\ddag).
    Immediately, $\plan' = \plan{z}$ also satisfies (\dag) and (\ddag).

    \medskip
    \noindent
    It is easy to see that, if there is $\plan$ satisfying (\dag) and (\ddag), then, ($\mathsection$) for every model $\model$ and $j \in [1,|\plan|-1]$, $\truthset{\model}{\A\psi} = \S$ implies $\truthset{\model}{\chi_{\plan_j}} \subseteq \truthset{\model}{\psi_{\plan_{(j+1)}}}$.


    \medskip
    \noindent
    Let us now resume with the main proof.
    Let $(x,y) \in \C(\varphi,\psi)$ and $\model$ be any model.
    The result is direct if $\truthset{\model}{\varphi \land \A\psi} = \emptyset$.
    Thus, consider $\truthset{\model}{\varphi \land \A\psi} \neq \emptyset$; i.e., s.t.\ $\truthset{\model}{\varphi \land \A\psi} = \S$.
    From ($*$), we know that exists a sequence $\plan$ of indices in $[1,n]$ that satisfies (\dag) and (\ddag).
    Then, for all $j \in [1,|\plan|-1]$, $\truthset{\model}{\chi_{\plan_j}} \subseteq \truthset{\model}{\psi_{\plan_{(j+1)}}}$.
    Using \Cref{prop:composition}, $\truthset{\model}{\varphi \land \A\psi} \subseteq \bigcap_{j=1}^{|\plan|} \truthset{\model}{\kh(\psi_{\plan_j},\chi_{\plan_j})} \subseteq \truthset{\model}{\kh(\psi_x,\chi_y)}$.
\end{proof}

\begin{proposition} \label{prop:alternative:composition}
    Let
        $\varphi = \kh(\psi_1,\chi_1) \land \dots \land \kh(\psi_n,\chi_n)$ and
        $\psi$ be a formula;
        $\C(\varphi,\psi)$ is the smallest set s.t.:
        (1) for all $x \in [1,n]$, $(x,x) \in \C(\varphi,\psi)$; and
        (2) if
            $\set{(x,y_0),(y_1,z)} \subseteq \C(\varphi,\psi)$ and
            $\unsat(\set{\psi, \chi_{y_0}, \lnot \psi_{y_1}})$, then,
            $(x,z) \in \C(\varphi,\psi)$.
\end{proposition}

\begin{algorithm}[!t]
    \caption{\sc Plans}
    \label{alg:plans}
    \begin{algorithmic}[1]
    \Require
        $\varphi^{+}$ is as in \Cref{def:compatible}
    \Function{Plans}{$\varphi^{+},\psi$}
        \State $\C \gets \Call{False}{n,n}$
            \Comment{$\C$ is an $n \times n$ matrix whose entries are all set to $\bot$ (false)}
        \For{$x \gets 1$ {\bf to} $n$}
           \State $\C(x,x) \gets \top$
        \EndFor
        \For{$x \gets 1$ {\bf to} $n$}
            \Comment{compute all $(x,z) \in \C(\varphi^{+},\psi)$}
            \For{$z \gets 1$ {\bf to} $n$}
                \For{$y_0 \gets 1$ {\bf to} $n$}
                    \For{$y_1 \gets 1$ {\bf to} $n$}
                        \State $\C(x,z) \gets
                            \C(x,z) \text{ or (}
                            \C(x,y_0) \text{ and}
                            \C(y_1,z) \text{ and }
                            \text{not} \sat(\set{\psi,\chi_{y_0}, \lnot \psi_{y_1}}))$
                    \EndFor
                \EndFor
            \EndFor
        \EndFor
        \State \Return $\setof{(x,y)}{\C(x,y)=\top}$
    \EndFunction
    \Ensure
        $\Call{Plans}{\varphi^{+},\psi} = \C(\varphi^{+},\psi)$ 
    \end{algorithmic}
\end{algorithm}

The function {\sc Plans} in \Cref{alg:plans} can be used to compute the set $\C(\varphi,\psi)$ in \Cref{def:composition}.
This function looks into whether a pair of indices belongs to this set using the result in \Cref{prop:alternative:composition}.

\begin{figure}[t!]
    \noindent
    \begin{minipage}[t]{.245\textwidth}
        \begin{center}
            \begin{tabular}{c|c|c|c}
                    & $\chi_1$ & $\chi_2$ & $\chi_3$
                \tabularnewline
                \hline
                $\psi_1$ & $\top$ & $\bot$ & $\bot$
                \tabularnewline
                \hline
                $\psi_2$ & $\bot$ & $\top$ & $\bot$
                \tabularnewline
                \hline
                $\psi_3$ & $\bot$ & $\bot$ & $\top$
                \tabularnewline
                \hline
            \end{tabular}
            \caption*{initial step}
        \end{center}
    \end{minipage}
    \begin{minipage}[t]{.245\textwidth}
        \begin{center}
            \begin{tabular}{c|c|c|c}
                    & $\chi_1$ & $\chi_2$ & $\chi_3$
                \tabularnewline
                \hline
                $\psi_1$ & $\top$ & {\color{red}$\top$} & $\bot$
                \tabularnewline
                \hline
                $\psi_2$ & $\bot$ & $\top$ & $\bot$
                \tabularnewline
                \hline
                $\psi_3$ & $\bot$ & $\bot$ & $\top$
                \tabularnewline
                \hline
            \end{tabular}
            \caption*{$\begin{array}{cccc} x &= 1, & y_0 &= 1 \\ z &= 2, & y_1 &= 2\end{array}$}
        \end{center}
    \end{minipage}
    \begin{minipage}[t]{.245\textwidth}
        \begin{center}
            \begin{tabular}{c|c|c|c}
                    & $\chi_1$ & $\chi_2$ & $\chi_3$
                \tabularnewline
                \hline
                $\psi_1$ & $\top$ & $\top$ & {\color{red}$\top$}
                \tabularnewline
                \hline
                $\psi_2$ & $\bot$ & $\top$ & $\bot$
                \tabularnewline
                \hline
                $\psi_3$ & $\bot$ & $\bot$ & $\top$
                \tabularnewline
                \hline
            \end{tabular}
            \caption*{$\begin{array}{cccc} x &= 1, & y_0 &= 2 \\ z &= 3, & y_1 &= 3\end{array}$}
        \end{center}
    \end{minipage}
    \begin{minipage}[t]{.245\textwidth}
        \begin{center}
            \begin{tabular}{c|c|c|c}
                    & $\chi_1$ & $\chi_2$ & $\chi_3$
                \tabularnewline
                \hline
                $\psi_1$ & $\top$ & $\top$ & $\top$
                \tabularnewline
                \hline
                $\psi_2$ & $\bot$ & $\top$ & {\color{red}$\top$}
                \tabularnewline
                \hline
                $\psi_3$ & $\bot$ & $\bot$ & $\top$
                \tabularnewline
                \hline
            \end{tabular}
            \caption*{$\begin{array}{cccc} x &= 2, & y_0 &= 2 \\ z &= 3, & y_1 &= 3\end{array}$}
        \end{center}
    \end{minipage}

    \vspace{7pt}
    \caption{A Run of {\sc Plans} for $\varphi = \kh(p,p \land q) \land \kh(q, r) \land \kh(r \lor s, t)$ and $\psi = \top$.}
    \label{fig:plans}
\end{figure}

\begin{example}
    Let $\varphi = \kh(p,p \land q) \land \kh(q, r) \land \kh(r \lor s, t)$ and $\psi = \top$;
    in this case we have: $\psi_1 = p$, $\chi_1 = p \land q$, $\psi_2 = q$, $\chi_2 = r$, $\psi_3 = r \lor s$, and ${\chi_3 = t}$.
    We can easily verify that $\C(\varphi,\psi) = \set{(1,1)(1,2)(1,3)(2,2)(2,3)(3,3)}$.
    Indeed, in the initial step we get ${\C_0 = \set{(1,1)(2,2)(3,3)}}$.
    The pairs of indices correspond to those of the pre/post conditions of the subformulas $\kh(\psi_i,\chi_i) \in \sforms(\varphi)$.
    Then, since we have $\set{(1,1)(2,2)} \subseteq \C_0$, $\unsat(\set{\chi_1,\lnot\psi_2})$, and $\unsat(\set{\chi_2,\lnot\psi_3})$, it follows that  $\C_1 = \C_0 \cup \set{(1,2)(2,3)}$.
    The new pairs of indices can intuitively be taken as the formulas $\kh(\psi_1,\chi_2)$ and $\kh(\psi_2,\chi_3)$.
    In this case, note the connection between $\kh(\psi_1,\chi_2)$ and ${(\kh(\psi_1,\chi_1)\land\A(\chi_1 \to \psi_2)\land\kh(\psi_2,\chi_2)) \to \kh(\psi_1,\chi_2)}$, and $\kh(\psi_2,\chi_3)$ and $(\kh(\psi_2,\chi_2)\land\A(\chi_2 \to \psi_3)\land\kh(\psi_3,\chi_3)) \to \kh(\psi_2,\chi_3)$.
    Finally, since we have $(1,2) \in \C_2$ and $\unsat(\set{\chi_2,\lnot\psi_3})$, then $\C_2 = \C_1 \cup \set{(1,3)}$.
    The justification for the pair $(1,3)$ is similar to the one just offered.
    In \Cref{fig:plans} we illustrate a run of {\sc Plans} which computes this set (only the steps in which the matrix is updated are shown).
\end{example}


The composition of formulas $\kh(\psi,\chi)$ and $\kh(\psi',\chi')$ has an impact if we wish to add a formula $\lnot\kh(\psi'',\chi'')$ into the mix.
The reason for this is that witness plans $\plan$ and $\plan'$ for $\kh(\psi,\chi)$ and $\kh(\psi',\chi')$, respectively, yield a witness plan $\plan'' = \plan\plan'$ for $\kh(\psi,\chi')$.
In adding $\lnot\kh(\psi'',\chi'')$ we need to ensure $\plan''$ is not a witness for $\kh(\psi'',\chi'')$, as such a plan renders $\lnot\kh(\psi'',\chi'')$ unsatisfiable.
We make these ideas precise in the definition of \emph{compatible} below.

\begin{definition} \label{def:compatible}
    Let $\varphi^{+}$ and $\varphi^{-}$ be formulas s.t.:
        $\depth(\varphi^{+}) = 1$ and $ \depth(\varphi^{-}) = 1$;
        $\varphi^{+}= \kh(\psi_1,\chi_1) \land \dots \land \kh(\psi_n,\chi_n)$; and 
        $\varphi^{-}= \lnot \kh(\psi_1',\chi_1') \land \dots \land \lnot \kh(\psi_m',\chi_m')$.
    Moreover, let $I, J \subseteq [1,n]$ be as in \Cref{prop:sat+} and $\psi = \bigwedge_{i \in I}\lnot \psi_i$.
    We say that $\varphi^{+}$ and $\varphi^{-}$ are \emph{compatible} iff the following conditions are met:
        \begin{enumerate}[(1)]
            \item $\sat(\psi)$;
            \item for all
                $\kh(\psi'_{k'},\chi'_{k'}) \in \sforms(\varphi^{-})$,
                \begin{enumerate}[(a)]
                    \item $\sat(\set{\psi, \psi'_{k'}, \lnot \chi'_{k'}})$; and
                    \item for all $(x,y) \in \C(\varphi^{+}, \psi)$, \\
                        \hspace*{\fill}
                        if
                            $x \notin J$ and
                            $\unsat(\set{\psi, \psi'_{k'}, \lnot \psi_x})$,
                        then,
                            $\sat(\set{\psi, \chi_y,\lnot \chi'_{k'}})$.
                \end{enumerate}
        \end{enumerate}
\end{definition}

\Cref{def:compatible} 
aims to single out the conditions under which the formulas $\varphi^{+}$ and $\varphi^{-}$ can be jointly satisfied.
Intuitively, (1) tells us $\varphi^{+}$ must be individually satisfied (cf.\ \Cref{prop:sat+}).
In turn, (2.a) tells us $\varphi^{-}$ must be individually satisfied (cf.\ \Cref{prop:sat-}),
while (2.b) tells us $\varphi^{+}$ and $\varphi^{-}$ can be satisfied together if no composition of subformulas in $\varphi^{+}$ contradicts a subformula in $\varphi^{-}$.
Such a contradiction would originate only as a result of strengthening the precondition and/or weakening the postcondition of a composition of subformulas in $\varphi^{+}$, in a way such that they would result in the opposite of a subformula in $\varphi^{-}$.
\Cref{prop:compatible} states that the conditions in \Cref{def:compatible} 
guarantee the satisfiability of a combination of $\varphi^+$ and~$\varphi^-$.

\begin{proposition} \label{prop:compatible}
    It follows that $\varphi^{+}$ and $\varphi^{-}$ are compatible iff $\sat(\varphi^{+} \land \varphi^{-})$.
\end{proposition}
\begin{proof}
    ($\Rightarrow$)
    Suppose that $\varphi^{+}$ and $\varphi^{-}$ are compatible.
    Let
        $\model = \tup{\S, \R, \V}$ be s.t.
        $\S$ contains all valuations that make $\psi$ true; and
        $\R = \setof{\R_{a_k}}{a_k \in \ACT}$ where
        \begin{align*}
            \R_{a_k} & =
                \begin{cases}
                    {\truthset{\model}{\psi_k} \times \truthset{\model}{\chi_k}}
                        & \text{if } k \in K \\
                    \emptyset
                        & \text{otherwise},
                \end{cases}
        \end{align*}
    for $K = [1,n] \setminus (I \cup J)$.
    From (1), we know $\S \neq \emptyset$.
    It is not difficult to see that $\tset{\varphi^{+}}^{\model} = \S$ (cf.\ \Cref{prop:sat+}).
    The proof is concluded if $\tset{\varphi^{-}}^{\model} = \S$.
    We proceed by contradiction.
    Let $k' \in [1,m]$ be s.t.\ $\truthset{\model}{\kh(\psi'_{k'},\chi'_{k'})} = \S$; i.e., ($*$) exists $\plan \in \PLANS$ s.t. $\truthset{\model}{\psi'_j} \subseteq \stexec(\plan)$ and $\R_{\plan}(\truthset{\model}{\psi'_j}) \subseteq \truthset{\model}{\chi'_j}$.
    We consider the following cases.

    \begin{description}
        \item[\textnormal{($\plan = \varepsilon$)}]
            From (2.a), we know ${\truthset{\model}{\psi'_{k'} \land \lnot \chi'_{k'}} \neq \emptyset}$; i.e.,
            $\truthset{\model}{\psi'_{k'}} \nsubseteq \truthset{\model}{\chi'_{k'}}$.
            This implies ${\truthset{\model}{\psi'_{k'}} = \R_{\varepsilon}(\truthset{\model}{\psi'_{k'}}) \nsubseteq \truthset{\model}{\chi'_{k'}}}$.

        \item[\textnormal{($\plan \neq \varepsilon$ and $\plan = a_{k_1},\dots,a_{k_{|\plan|}}$ with $k_j \in K$ and $j \in [1,|\plan|]$)}]
            In this case we have:
            \begin{enumerate}[(a)]
                \item $\emptyset \neq \truthset{\model}{\psi'_{k'}} \subseteq \stexec(\plan) \subseteq \stexec(a_{k_1}) = \truthset{\model}{\psi_{k_1}}$;
                \item 
                    $\truthset{\model}{\chi_{k_j}} = \R_{a_{k_j}}(\truthset{\model}{\psi_{k_j}}) \subseteq \truthset{\model}{\psi_{k_{(j+1)}}}$; and
                \item $\truthset{\model}{\chi_{k_{|\plan|}}} = \R_{\plan}(\truthset{\model}{\psi'_{k'}}) \subseteq \truthset{\model}{\chi'_{k'}}$.
            \end{enumerate}
            Since $\S$ contains all valuations that make $\psi$ true; from (a)--(d) we get:
            \begin{enumerate}[(a)]
                \setcounter{enumi}{3}
                \item $\unsat(\set{\psi, \psi'_{k'}, \lnot \psi_{k_1}})$ --from (a);
                \item $\unsat(\set{\psi, \chi_{k_j}, \lnot \psi_{k_{(j+1)}}})$ --from (b);
                \item $\unsat(\set{\psi, \chi_{k_{|\plan|}}, \lnot \chi'_k})$ --from (c).
            \end{enumerate}
            From (e) and $\plan$, we obtain a sequence $k_1 \dots k_{|\plan|}$ that satisfies the conditions (\dag) and (\ddag) in the proof of \Cref{prop:chain}.
            Then, $(k_1,k_{|\plan|}) \in \C(\varphi^{+},\psi)$.
            From (a) and (2.a), $k_1 \notin J$.
            We are in an impossible situation:
                $(k_1,k_{|\plan|}) \in \C(\varphi^{+},\psi)$;
                $k_1 \notin J$; and
                $\unsat(\set{\psi, \chi_{k_{|\plan|}}, \lnot \chi'_k})$.
            This contradicts (2.b); meaning that $\varphi^{+}$ and $\varphi^{-}$ are not compatible.        
        \item[\textnormal{($\plan$ is none of the above)}]
            It is clear that $\truthset{\model}{\psi'_{k'}} \nsubseteq \stexec(\plan)$.
    \end{description}
    In all the cases above we have: $\truthset{\model}{\psi'_{k'}} \nsubseteq \stexec(\plan)$ or $\R_{\plan}(\truthset{\model}{\psi'_{k'}}) \nsubseteq \truthset{\model}{\chi'_{k'}}$; i.e., $\truthset{\model}{\kh(\psi'_{k'},\chi'_{k'})} = \emptyset$, a contradiction. 
    Then, $\truthset{\model}{\varphi^{-}} = \S$; and so $\sat(\varphi^{+} \land \varphi^{-})$.

    \medskip
    \noindent
    ($\Leftarrow$)
    Suppose $\sat(\varphi^{+} \land \varphi^{-})$; i.e.,
        exists ($\dagger$) $\model$ s.t.\
            $\tset{\varphi^{+} \land \varphi^{-}}^{\model} = \S$.
    From ($\dagger$) we get $\truthset{\model}{\varphi^{+}} = \S$.
    Using \Cref{prop:nts:universal}, we get $\truthset{\model}{\A\psi} = \S$.
    This establishes (1).
    The proof of (2.a) is by contradiction.
    Let $\kh(\psi'_{k'},\chi'_{k'}) \in \sforms(\varphi^{-})$ be s.t.\
        $\unsat(\set{\psi, \psi'_{k'}, \lnot \chi'_{k'}})$.
    Then, $\truthset{\model}{\psi'_{k'}} \subseteq \truthset{\model}{\chi'_{k'}}$.
    Choosing $\plan=\epsilon$, we obtain $\truthset{\model}{\kh(\psi'_{k'},\chi'_{k'})} = \S$.
    This contradicts ${\truthset{\model}{\varphi^{-}} = \S}$.
    The proof of (2.b) is also by contradiction.
    Let
        $\kh(\psi'_{k'},\chi'_{k'}) \in \sforms(\varphi^{-})$, 
        ($*$) ${(x,y) \in \C(\varphi^{+},\psi)}$,
        (\dag) $\unsat(\set{\psi,\psi'_{k'}, \lnot \psi_x})$, and
        (\ddag) $\unsat(\set{\psi, \chi_y, \lnot \chi'_{k'}})$.
    From (\dag) and (\ddag),
        $\truthset{\model}{\psi'_{k'}} \subseteq \truthset{\model}{\psi_x}$ and
        $\truthset{\model}{\chi_y} \subseteq \truthset{\model}{\chi'_{k'}}$.
    At the same time, from ($*$) and \Cref{prop:chain},
        $\S = \truthset{\model}{\varphi^{+}} \subseteq \truthset{\model}{\kh(\psi_x,\chi_y)}$.
    Then, using \Cref{prop:composition},
        $\truthset{\model}{\kh(\psi'_j,\chi'_j)} = \S$.
    This also contradicts $\truthset{\model}{\varphi^{-}} = \S$.
    Thus, $\varphi^{+}$ and $\varphi^{-}$ are compatible.
\end{proof}

Having at hand the result in \Cref{prop:compatible}, we proceed to define an algorithm for checking the satisfiability of compatible formulas $\varphi^{+}$ and $\varphi^{-}$.
This is done in two stages.
In the first stage, we build the set $\C(\varphi^{+},\psi)$, where $\psi$ is the conjunction of the negation of the precondition of the `global' subformulas in~$\varphi^{+}$.
This task is encapsulated in the function {\sc Plans} in~\Cref{alg:plans}.
Notice that the set $\C(\varphi^{+},\psi)$ corresponds to a matrix which is computed using the result in \Cref{prop:alternative:composition}.
The second stage is encapsulated in the function {\sc Compatible} in~\Cref{alg:compatible}.
In this function, lines $2$ and $3$ check condition (1) in~\Cref{def:compatible}, 
i.e., whether~$\varphi^{+}$ is individually satisfiable, by verifying the joint satisfiability of the `global' subformulas in $\varphi^{+}$ (cf.~\Cref{alg:sat+}).
In turn, lines $4$ to $6$ in {\sc Compatible} check condition (2.a) of~\Cref{def:compatible},
i.e., whether $\varphi^{-}$ is individually satisfiable, by verifying the individual satisfiability of the subformulas in $\varphi^{+}$ (cf.~\Cref{alg:sat-}).
Lastly, in lines $7$ to $18$ in {\sc Compatible}, we check whether the result of composing subformulas in $\varphi^{+}$ contradicts any of the subformulas in $\varphi^{-}$.
We carry out this task by making use of the result of the function {\sc Plans} which computes such compositions.

\algdef{SE}{Begin}{End}{\textbf{begin}}{\textbf{end}}
\begin{algorithm}[!t]
   \caption{\sc Compatible}
     \label{alg:compatible}
     \begin{algorithmic}[1]
    \Require
            $\varphi^{+}$ and $\varphi^{-}$ are as in \Cref{def:compatible}
    \Function{Compatible}{$\varphi^{+},\varphi^{-}$}
        \State $\Psi \gets \setof{\lnot \psi_i}{i \in \Call{Global}{\varphi^{+}}}$
        \State $r \gets \sat(\Psi)$
            \Comment{check for condition (1)}
        \For{$k' \gets 1$ {\bf to} $m$}
            \Comment{check for condition (2.a)}
            \State $r \gets r$ and $\sat(\Psi \cup \{\psi'_{k'},\lnot\chi'_{k'}\})$
        \EndFor
        \State $\C \gets \Call{Plans}{\varphi^{+},\bigwedge\Psi}$
        \For{$k' \gets 1$ {\bf to} $m$}
            \Comment{check for condition (2.b)}
            \For{$x \gets 1$ {\bf to} $n$}
                \For{$y \gets 1$ {\bf to} $n$}
                    \If{$(x,y) \in \C$ and $\sat(\Psi \cup \{\psi_x\})$ and not $\sat(\Psi \cup \{\psi'_{k'}, \lnot \psi_x\})$}
                        \State
                            $r \gets
                                r \text{ and }
                                \sat(\Psi \cup \set{\chi_y, \lnot \chi'_{k'}})$
                    \EndIf
                \EndFor
            \EndFor
        \EndFor
        \State \Return $r$
    \EndFunction
    \Ensure
            $\Call{Compatible}{\varphi^{+},\varphi^{-}}$ iff
            $\varphi^{+}$ and $\varphi^{-}$ are compatible
    \end{algorithmic}
\end{algorithm}

Notice that the function {\sc Compatible} in \Cref{alg:compatible} makes a polynomial number of calls to a propositional $\sat$ solver.
From this fact, we get the following result.

\begin{proposition}\label{prop:alg:compatible}
    Let $\varphi^{+}$, $\varphi^{-}$ be as in \Cref{def:compatible}; it follows that \Cref{alg:compatible} solves $\sat(\varphi^{+} \land \varphi^{-})$ and is in $\Poly^{\NP}$ (i.e., $\Delta^\Poly_2$ in \PH).
\end{proposition}

\begin{proof}
    By~\Cref{prop:compatible} we get that the function {\sc Compatible} in \Cref{alg:compatible} solves  $\sat(\varphi^{+} \land \varphi^{-})$. 
    Moreover, it makes a polynomial number of calls to a $\sat$ solver for formulas of modal depth $0$. 
    Thus, it runs in polynomial time with access to a $\sat$ oracle.
    Therefore, $\sat(\varphi^{+} \land \varphi^{-})$ is in $\Poly^{\NP}$, i.e., in $\Delta^{\Poly}_2$. 
\end{proof}

\Cref{prop:alg:compatible} is the final step we need to reach the main result of our work.

\begin{theorem}\label{theo:complexity}
    The satisfiability problem for $\KHlogic$ is in $\NP^\NP$ (i.e., $\Sigma^\Poly_2$ in \PH). 
\end{theorem}
\begin{proof}
    Let $\varphi$ be a $\KHlogic$-formula.
    By \Cref{alg:form2form1}, we can obtain, in polynomial time, a formula $\varphi'=\varphi_0 \land (\A{p_1} \leftrightarrow \kh(\psi_1,\chi_1)) \land \dots \land (\A{p_n} \leftrightarrow \kh(\psi_n,\chi_n))$ in leaf normal form such that $\varphi\equiv_{\sat}\varphi'$. 
    We know  $\depth(\varphi_0) = 0$ and $\depth(\kh(\psi_i,\chi_i)) = 1$.
    Let $Q = \set{q_1 \dots q_m} \subseteq \PROP$ be the set of proposition symbols in $\varphi'$.
    To check $\sat(\varphi')$, we start by guessing a propositional assignment $v: Q \to \set{0,1}$ that makes $\varphi_0$ true. 
    Then, we define sets
        $P^{+} = \setof{i}{v(p_i) = 1}$ and
        $P^{-} = \setof{i}{v(p_i) = 0}$,
        from which we build formulas
    \begin{align*}
        \varphi^{+}
            & = \textstyle \bigwedge_{i \in P^{+}}\kh(\psi_i,\chi_i) &
        \varphi^{-}
            & = \textstyle \left(\bigwedge_{i \in P^{-}}\lnot\kh(\psi_i,\chi_i)\right) \land \lnot\kh(\varphi_0,\bot)
    \end{align*}
    (recall that $\lnot\kh(\varphi_0,\bot) = \lnot\A\lnot\varphi_0 = \E\varphi_0$.)
    Finally, we use~\Cref{alg:compatible} to check $\sat(\varphi^+\wedge\varphi^-)$.
    Since \Cref{alg:compatible} is in $\Poly^{\NP}$ (\Cref{prop:alg:compatible}), the whole process is in $\NP^{\NP}$.
\end{proof}

We conclude this section with an example of how to check the satisfiability of a formula using the procedure in the proof of \Cref{theo:complexity}.

\begin{example}
Let $\psi = \kh(p \wedge q, r \wedge t) \vee \kh(p, r)$.
By applying \Cref{alg:form2form1}, we get $(k_1 \vee k_2) \wedge (\A k_1 \leftrightarrow  \kh(p \wedge q,  r \wedge t)) \wedge (\A k_2 \leftrightarrow \kh(p, r))$.
Suppose that we set $k_1$ to true and $k_2$ to false.
Based on this assignment, we build formulas $\varphi^+ = \kh(p \wedge q,  r \wedge t)$ and $\varphi^- = \neg \kh(p, r) \wedge \lnot \kh(k_1 \wedge \neg k_2,\bot)$.
Using \Cref{alg:compatible}, we can check that they are not compatible (and hence not satisfiable; we have $\sat(p \wedge q)$ and $\unsat(\set{(p \wedge q), \neg p})$ but not $\sat(\set{r\wedge t, \lnot r})$).
However, if we set both $k_1$ and $k_2$ to true, then, $\varphi^+ =\kh(p \wedge q, r \wedge t) \wedge  \kh(p, r)$ and $\varphi^- = \lnot\kh(k_1 \wedge k_2,\bot)$.
In this case, \Cref{alg:compatible} returns they are compatible, and thus satisfiable.

\end{example}

\section{Final Remarks}
\label{sec:final}
We provided a satisfiability-checking procedure for $\KHlogic$, the `knowing how' logic with linear plans from~\cite{Wang15lori,Wang2016},  obtaining a $\Sigma^P_2$ upper bound.
Although not a tight bound (as the best lower bound known is \NP), we argue this is an interesting result, as our bound is (unless \PH collapses) below the \PSPACE-complete complexity of model-checking~\cite{DF23}. We argue that, this unusual situation is a consequence of that in model-checking the full expressive power is exploited, while here we showed that plans are almost irrelevant for the satisfiability of a formula. 

Interestingly also, our procedure uses a polynomial transformation into a normal form without nested modalities, and calls to an \NP oracle (i.e., to a propositional $\sat$ solver). 
It is well-known that modern $\sat$ solvers are able to efficiently deal with large formulas (having millions of variables),  and usually support the exploration of the solution state space.  Thus, the ideas presented in this paper can be used to implement a $\sat$ solver for knowing-how logics relying on modern propositional $\sat$ solving tools.  We consider this as part of the future work to undertake. Also, we would like to obtain a tight bound for the satisfiability problem. In this regard, we will explore the possibility of providing a reduction from the problem of checking the truth of Quantified Boolean Formula (TQBF) with a single $\exists\forall$ quantification pattern (called $\Sigma_2\mathsf{Sat}$ in~\cite{AroraB09}), which is known to be $\Sigma^P_2$-complete. 

\paragraph*{\bf Acknowledgments.} We thank the reviewers for their valuable comments.
Our work is supported by 
the Laboratoire International Associ\'e SINFIN,
the EU Grant Agreement 101008233 (MISSION),
the
  ANPCyT projects
    PICT-2019-03134,
    PICT-2020-3780,
    PICT-2021-00400,
    PICT-2021-00675, and
    PICTO-2022-CBA-00088,
and the
  CONICET projects
    PIBAA-28720210100428CO,
    PIBAA-28720210100165CO, and
    PIP-11220200100812CO.

\bibliographystyle{splncs04}
\bibliography{references}

\begin{thebibliography}{10}
\providecommand{\url}[1]{\texttt{#1}}
\providecommand{\urlprefix}{URL }
\providecommand{\doi}[1]{https://doi.org/#1}

\bibitem{AFSVQ21}
Areces, C., Fervari, R., Saravia, A.R., Vel{\'{a}}zquez{-}Quesada, F.R.:
  Uncertainty-based semantics for multi-agent knowing how logics. In: 18th
  Conference on Theoretical Aspects of Rationality and Knowledge ({TARK} 2021).
  {EPTCS}, vol.~335, pp. 23--37. Open Publishing Association (2021)

\bibitem{AroraB09}
Arora, S., Barak, B.: Computational Complexity: A Modern Approach. Cambridge
  University Press, 1st edn. (2009)

\bibitem{Baltag16}
Baltag, A.: To know is to know the value of a variable. In: Advances in Modal
  Logic ({AiML} 2016). vol.~11, pp. 135--155. College Publications (2016)

\bibitem{Cormen22}
Cormen, T., Leiserson, C., Rivest, R.L., Stein, C.: Introduction to Algorithms.
  {MIT} Press, 4th edn. (2022)

\bibitem{DF23}
Demri, S., Fervari, R.: Model-checking for ability-based logics with
  constrained plans. In: 37th {AAAI} Conference on Artificial Intelligence
  ({AAAI} 2023). pp. 6305--6312. {AAAI} Press (2023)

\bibitem{HBEL}
van Ditmarsch, H., Halpern, J.Y., van~der Hoek, W., Kooi, B. (eds.): Handbook
  of Epistemic Logic. College Publications (2015)

\bibitem{FWvD15}
Fan, J., Wang, Y., van Ditmarsch, H.: Contingency and knowing whether. The
  Review of Symbolic Logic  \textbf{8},  75--107 (2015)

\bibitem{fantl2012introduction}
Fantl, J.: Knowledge how. In: The Stanford Encyclopedia of Philosophy.
  Metaphysics Research Lab, Stanford University, spring 2021 edn. (2021)

\bibitem{FervariHLW17}
Fervari, R., Herzig, A., Li, Y., Wang, Y.: Strategically knowing how. In: 26th
  International Joint Conference on Artificial Intelligence ({IJCAI} 2017). pp.
  1031--1038. International Joint Conferences on Artificial Intelligence (2017)

\bibitem{FervariVQW22}
Fervari, R., Vel\'azquez-Quesada, F.R., Wang, Y.: Bisimulations for knowing how
  logics. The Review of Symbolic Logic  \textbf{15}(2),  450--486 (2022)

\bibitem{GorankoP92}
Goranko, V., Passy, S.: Using the universal modality: Gains and questions.
  Journal of Logic and Computation  \textbf{2}(1),  5--30 (1992)

\bibitem{GW16}
Gu, T., Wang, Y.: ``{K}nowing value'' logic as a normal modal logic. In:
  Advances in Modal Logic ({AiML} 2016). vol.~11, pp. 362--381. College
  Publications (2016)

\bibitem{Herzig15}
Herzig, A.: Logics of knowledge and action: critical analysis and challenges.
  Autonomous Agents and Multi-Agent Systems  \textbf{29}(5),  719--753 (2015)

\bibitem{HerzigT06}
Herzig, A., Troquard, N.: Knowing how to play: uniform choices in logics of
  agency. In: 5th International Joint Conference on Autonomous Agents and
  Multiagent Systems ({AAMAS} 2006). pp. 209--216. {ACM} (2006)

\bibitem{Hintikka:1962}
Hintikka, J.: Knowledge and Belief. Cornell University Press (1962)

\bibitem{wiebeetal:2003}
van~der Hoek, W., Lomuscio, A.: Ignore at your peril -- towards a logic for
  ignorance. In: 2nd International Conference on Autonomous Agents and
  MultiAgent Systems ({AAMAS} 2003). pp. 1148--1149. {ACM} (2003)

\bibitem{JamrogaA07}
Jamroga, W., {\AA}gotnes, T.: Constructive knowledge: what agents can achieve
  under imperfect information. Journal of Applied Non Classical Logics
  \textbf{17}(4),  423--475 (2007)

\bibitem{Les00}
Lesp{\'{e}}rance, Y., Levesque, H.J., Lin, F., Scherl, R.B.: Ability and
  knowing how in the situation calculus. Studia Logica  \textbf{66}(1),
  165--186 (2000)

\bibitem{Li17}
Li, Y.: Stopping means achieving: A weaker logic of knowing how. Studies in
  Logic  \textbf{9}(4),  34--54 (2017)

\bibitem{Li23}
Li, Y.: Tableaux for the logic of strategically knowing how. In: 19th
  Conference on Theoretical Aspects of Rationality and Knowledge ({TARK} 2023).
  EPTCS, vol.~379, pp. 379--391. Open Publishing Association (2023)

\bibitem{LiWang17}
Li, Y., Wang, Y.: Achieving while maintaining: {A} logic of knowing how with
  intermediate constraints. In: 7th Indian Conference on Logic and Its
  Applications ({ICLA} 2017). pp. 154--167. LNCS, Springer (2017)

\bibitem{Li17bis}
Li, Y.: Knowing what to do: a logical approach to planning and knowing how.
  Ph.D. thesis, University of Groningen (2017)

\bibitem{Li21}
Li, Y.: Tableau-based decision procedure for logic of knowing-how via simple
  plans. In: 4th International Conference on Logic and Argumentation ({CLAR}
  2021). LNCS, vol. 13040, pp. 266--283. Springer (2021)

\bibitem{LiW21}
Li, Y., Wang, Y.: Planning-based knowing how: {A} unified approach. Artificial
  Intelligence  \textbf{296} (2021)

\bibitem{Mccarthy69}
McCarthy, J., Hayes, P.J.: Some philosophical problems from the standpoint of
  artificial intelligence. In: Machine Intelligence. pp. 463--502. Edinburgh
  University Press (1969)

\bibitem{Moore85}
Moore, R.: A formal theory of knowledge and action. In: Formal Theories of the
  Commonsense World. Ablex Publishing Corporation (1985)

\bibitem{NaumovT18}
Naumov, P., Tao, J.: Second-order know-how strategies. In: 17th International
  Conference on Autonomous Agents and MultiAgent Systems ({AAMAS} 2018). pp.
  390--398. {ACM} (2018)

\bibitem{Naumov2018a}
Naumov, P., Tao, J.: Together we know how to achieve: An epistemic logic of
  know-how. Artificial Intelligence  \textbf{262},  279--300 (2018)

\bibitem{Smith&Weld98}
Smith, D.E., Weld, D.S.: Conformant graphplan. In: 15th National Conference on
  Artificial Intelligence and 10th Innovative Applications of Artificial
  Intelligence Conference ({AAAI}/{IAAI} '98). pp. 889--896. {AAAI} Press/The
  {MIT} Press (1998)

\bibitem{Stock76}
Stockmeyer, L.J.: The polynomial-time hierarchy. Theoretical Computer Science
  \textbf{3}(1),  1--22 (1976)

\bibitem{Wang15lori}
Wang, Y.: A logic of knowing how. In: 5th International Workshop on Logic,
  Rationality, and Interaction ({LORI} 2015). pp. 392--405. LNCS, Springer
  (2015)

\bibitem{Wang16}
Wang, Y.: Beyond knowing that: a new generation of epistemic logics. In: J.
  Hintikka on knowledge and game theoretical semantics, pp. 499--533. Springer
  (2018)

\bibitem{Wang2016}
Wang, Y.: A logic of goal-directed knowing how. Synthese  \textbf{195}(10),
  4419--4439 (2018)

\bibitem{XuW16}
Xu, C., Wang, Y., Studer, T.: A logic of knowing why. Synthese
  \textbf{198}(2),  1259--1285 (2021)

\end{thebibliography}




\end{document}